\documentclass[submission,copyright,creativecommons]{eptcs}
\providecommand{\event}{TARK 2017} % Name of the event you are submitting to
 \usepackage{underscore}           % Only needed if you use pdflatex.
\usepackage{amsmath,stmaryrd, graphicx, amssymb}
\usepackage{amsthm}

\graphicspath{ {} }

\newtheorem{proposition}{Proposition} [section]
\newtheorem{corollary}[proposition]{Corollary}

\newtheorem{lemma}[proposition]{Lemma}

\theoremstyle{definition}
\newtheorem {definition}[proposition]{Definition}

\title{Coalition and Group Announcement Logic}
\author{Rustam Galimullin \qquad\qquad Natasha Alechina
\institute{School of Computer Science\\
University of Nottingham\\
Nottingham, UK}
\email{\{rustam.galimullin, natasha.alechina\}@nottingham.ac.uk}
}
\def\titlerunning{Coalition and Group Announcement Logic}
\def\authorrunning{Rustam Galimullin \& Natasha Alechina}
\begin{document}
\maketitle

%\documentclass{acmart}
%\usepackage[utf8]{inputenc}
%\usepackage{amsmath,stmaryrd, graphicx, amssymb}

%\graphicspath{ {} }

%\newtheorem{proposition}{Proposition}
%\newtheorem{corollary}[proposition]{Corollary}
%\newtheorem{theorem}[proposition]{Theorem}
%\newtheorem{lemma}[proposition]{Lemma}
%\newdef{definition}{Definition}

%\acmConference[TARK'17]{Theoretical Aspects of Rationality and Knowledge}{July 2017}{Liverpool, UK}

%\begin{document}

%\title{Coalition and Group Announcement Logic}

%\author{Rustam Galimullin}
%\author{Natasha Alechina}
%\affiliation{School of Computer Science, University of Nottingham}
%\email{\{rustam.galimullin, natasha.alechina\}.nottingham.ac.uk}

\begin{abstract}
%Dynamic epistemic logics which model abilities of agents to make various announcements and influence each other's knowledge have been studied extensively in recent years. However, many difficult questions remain unanswered. In this paper we consider Coalition Announcement Logic combined with Group Announcement Logic, and make some progress towards settling questions of validity of various properties and of axiomatisation.
Dynamic epistemic logics which model abilities of agents to make various announcements and influence each other's knowledge have been studied extensively in recent years. Two notable examples of such logics are Group Announcement Logic and Coalition Announcement Logic. They allow us to reason about what groups of agents can achieve through joint announcements in non-competitive and competitive environments. In this paper, we consider a combination of these logics -- Coalition and Group Announcement Logic and provide its complete axiomatisation. Moreover, we partially answer the question of how group and coalition announcement operators interact, and settle some other open problems.
\end{abstract}

 %\begin{CCSXML}
%<ccs2012>
%<concept>
%<concept_id>10003752.10003790.10003793</concept_id>
%<concept_desc>Theory of computation~Modal and temporal logics</concept_desc>
%<concept_significance>500</concept_significance>
%</concept>
%</ccs2012>
%\end{CCSXML}

%\ccsdesc[500]{Theory of computation~Modal and temporal logics}
%\printccsdesc

%\terms{something general}
%\keywords{Dynamic epistemic logic, group announcement logic, coalition announcement logic}

%\maketitle

\section{Introduction}

To introduce the logics we will be working with in this paper, we start with an example loosely based on the one from \cite{renne09}.
Let us imagine that Ann, Bob, and Cath are travelling by train from Nottingham to Liverpool through Manchester. Cath was sound asleep all the way, and she has just woken up. She does not know whether the train passed Manchester, but Ann and Bob know that it has not. Now, if the train driver announces that the train is approaching Manchester, then Cath, as well as Ann and Bob, knows that they have not passed the city yet. To reason about changes in agents' knowledge after public announcements, we can use Public Announcement Logic (\(\mathbf{PAL}\)) \cite{plaza07}. Returning to the example, let us assume that the train driver does not announce anything, so that Cath is not aware of her whereabouts. Ann and Bob may tell her whether they passed Manchester.
%, or they may say some generally known truth (e.g. that it is either raining in Liverpool or not). 
In other words, Ann and Bob have an announcement that can influence Cath's knowledge. An extension of \(\mathbf{PAL}\), Group Announcement Logic (\(\mathbf{GAL}\)) \cite{agotnes10}, deals with the \emph{existence} of announcements by groups of agents that can achieve certain results. Now, let us assume that Ann does not want to disclose to Cath their whereabouts and Bob does, i.e. Ann and Bob have different goals. Then, it is clear that no matter what Ann says, the coalition of Bob and Cath can achieve the goal of Cath knowing that the train has not passed Manchester, that is, Bob can communicate this information to Cath. On the other hand, if Ann and Bob work together, then they have an announcement (for example, a tautology `It either rains in Liverpool or it doesn't'), such that whatever Cath says, she remains unaware of her whereabouts. For this type of strategic behaviour, another extension of \(\mathbf{PAL}\) -- Coalition Announcement Logic (\(\mathbf{CAL}\)) -- has been introduced in \cite{agotnes08}.

$\mathbf{CAL}$ joins two logical traditions: Dynamic Epistemic Logic, of which \(\mathbf{PAL}\) is a representative, and Coalition Logic (\(\mathbf{CL}\)) \cite{pauly02}. The latter allows us to reason about whether a coalition of agents has a strategy to achieve some goal, no matter what the agents outside of the coalition do. \(\mathbf{CL}\) essentially talks about concurrent games, and the actions that the agents execute are arbitrary actions (strategies in one-shot games). So, from this perspective, \(\mathbf{CAL}\) is a coalition logic with available actions restricted to public announcements. 

To the best of our knowledge, there is no complete axiomatisation of \(\mathbf{CAL}\) \cite{agotnes08, vanditmarsch12, agotnes14, agotnes16} or any other logic with coalition announcement operators. In this paper, we consider Coalition and Group Announcement Logic (\(\mathbf{CoGAL}\)), a combination of \(\mathbf{GAL}\) and \(\mathbf{CAL}\), which includes operators for both group and coalition announcements. The main result of this paper is a sound and complete axiomatisation of  \(\mathbf{CoGAL}\). As part of this result, we study the interplay between group and coalition announcement operators, and partially settle the question on their interaction that was stated as an open problem in \cite{vanditmarsch12, agotnes16}.

\section{Coalition and Group Announcement Logic} \label{sec:cogal}
\subsection {Syntax and Semantics}

Throughout the paper, let a finite set of agents \(A\), and a countable set of propositional variables \(P\) be given. The language of the logic is comprised of the language of classical propositional logic with added operators for agents' knowledge \(K_a \varphi\) (reads `agent \(a\) knows \(\varphi\)'), and public announcement \([\psi] \varphi\) (reads `after public announcement that \(\psi\), \(\varphi\) holds), group \([G] \varphi\) (`after any public announcement by group of agents \(G\), \(\varphi\) holds), and coalition announcements \([ \! \langle G \rangle \! ] \varphi\) (`for every public announcement by coalition of agents \(G\) there is an announcement by other agents \(A \setminus G\), such that \(\varphi\) holds after joint simultaneous announcement').

\begin{definition} (Language)
    The \emph{language of coalition and group announcement logic} \(\mathcal{L}_{CoGAL}\) is as follows:
    \begin{center}
        \(\varphi,\psi ::= p \mid \neg \varphi \mid (\varphi \wedge \psi) \mid K_a \varphi \mid [\varphi]\psi \mid [G]\varphi \mid [ \! \langle G \rangle \! ] \varphi\),
    \end{center}
    where \(p \in P\), \(a \in A\), \(G \subseteq A\), and all the usual abbreviations of propositional logic (such as \(\vee, \rightarrow, \leftrightarrow\)) and conventions for deleting parentheses hold. The dual operators are defined as follows: \(\widehat {K}_a \varphi \leftrightarrow \neg K_a \neg \varphi\), \(\langle \varphi \rangle \psi \leftrightarrow \neg [\varphi] \neg \psi\), \(\langle G \rangle \varphi \leftrightarrow \neg [ G ] \neg \varphi\), and \( \langle \! [ G ] \! \rangle \varphi \leftrightarrow \neg [ \! \langle G \rangle \! ] \neg \varphi\).  
Observe that  $\langle G \rangle \varphi$ means that $G$ has an announcement 
after which $\varphi$ holds, and $\langle \! [ G ] \! \rangle \varphi$ means 
that $G$ has an announcement such that after it is made simultaneously with 
any announcement by $A\setminus G$, $\varphi$ holds. The latter corresponds
to the Coalition Logic operator, but for announcements instead of arbitrary 
actions.

We define \(\mathcal{L}_{GAL}\) as the language without the operator \([ \! \langle G \rangle \! ]\), \(\mathcal{L}_{PAL}\) the language without \([G]\) as well, and \(\mathcal{L}_{EL}\) the purely epistemic language which in addition does not contain announcement operators \([\varphi]\).

\end{definition}

Next definition is needed for technical reasons in the formulation of infinite rules of inference in Definition \ref{def::axiomatisation}. We want the rules to work for a class of different types of premises. Ultimately, we require premises to be expressions of  depth \(n\) of the type \(\varphi_1 \rightarrow \square_1 (\varphi_2 \rightarrow \mathellipsis (\varphi_n \rightarrow \square_n \sharp) \mathellipsis)\), where \(\square_i\) is either \(K_a\) or \([\psi]\) for some \(a \in A\) and \(\psi \in \mathcal{L}_{CoGAL}\), atom \(\sharp\) denotes a placement of a formula to which a derivation is applied,  and some \(\varphi\)'s and \(\square\)'s can be omitted. This condition is captured succinctly by necessity forms originally introduced by Goldblatt in \cite{goldblatt}.

\begin{definition} (Necessity forms)
    Let \(\varphi \in \mathcal{L}_{CoGAL}\), then \emph{necessity forms} \cite{goldblatt} are inductively defined as follows:
    \begin{center}
        \(\eta ::= \sharp \mid \varphi \rightarrow \eta (\sharp) \mid K_a \eta (\sharp) \mid [\varphi] \eta (\sharp)\).
    \end{center}
    The atom \(\sharp\) has a unique occurrence in each necessity form. The result of the replacement of \(\sharp\) with \(\varphi\) in some \(\eta (\sharp)\) is denoted as \(\eta (\varphi)\). 
\end{definition}

Whereas formulas of coalition logic \cite{pauly02} are interpreted in game structures, formulas of \(\mathbf{CoGAL}\) are interpreted in epistemic models. Let us consider an example of such a model first. 
\begin{figure}[h]
    \centering
    \includegraphics {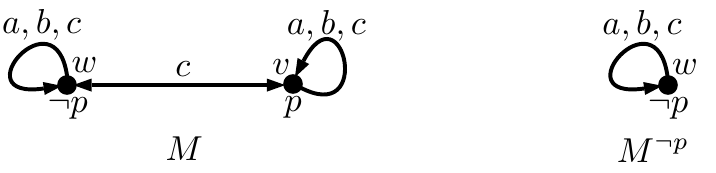}
    \caption{Train example}
    \label{fig::tark1}
\end{figure}
In Figure \ref{fig::tark1} there are three agents: \(a\) (Ann), \(b\) (Bob), and \(c\) (Cath). Let \(p\) denote the proposition that `The train has passed Manchester.' There are two states in the model \(M\): a state $w$ where
\(\neg p\) is true, and a state $v$ where \(p\) is true; and only one state in model \(M^{\neg p}\) which denotes $M$ updated by the announcement $\neg p$ (the process of updating the model is described below). % States represent possible worlds from agents' viewpoint. 
Let the $w$ be the actual state. Edges connect states that an agent cannot distinguish. In the actual state $w$ of $M$, Cath (agent \(c\)) does not know whether $p$ is true. Ann and Bob, on the contrary, know that $p$ is false. Now suppose that Bob announces that $\neg p$. This truthful public announcement `deletes' all the states where $p$ is true, and the corresponding epistemic indistinguishability relations; in this example, $v$ is `deleted,' and the resulting model is \(M^{\neg p}\). After this announcement Cath knows \(\neg p\), or, formally, \([\neg p] K_c \neg p\). In this paper, within group and coalition announcements, we only quantify over announcements of formulas of the type \(K_a \varphi\). If a group consists only of Cath, who does not know $\neg p$ and hence cannot announce $K_c \neg p$, the following holds in state $w$ of \(M\): \([c] (\neg K_c \neg p \wedge \neg K_c p)\), i.e. whatever \(c\) announces, she still does not know whether \(p\) after the announcement \footnote{For readability, we use $[c]$ rather than $[\{c\}]$ for singleton coalitions.}.  Also, Ann and Bob can remain silent (or announce a tautology \(\top\)) and preclude Cath from knowing that $\neg p$. In other words, there is announcement by their group such that after it is made, agent \(c\) does not know the value of \(p\): \(\langle \{a, b\} \rangle (\neg K_c \neg p \wedge \neg K_c p)\). Moreover, this holds whatever Cath announces at the same time: \(\langle \! [ \{a, b\} ] \! \rangle (\neg K_c \neg p \wedge \neg K_c p)\). On the other hand, a coalition consisting of Ann and Cath does not have such a power, since Bob can always announce that $\neg p$:
\(\neg \langle \! [ \{a, c\} ] \! \rangle (\neg K_c \neg p \wedge \neg K_c p)\), or, equally, \([ \! \langle \{a, c\} \rangle \! ] (K_c \neg p \vee K_c p)\).

Now, we provide formal definitions.

\begin{definition} (Epistemic model)
An \emph{epistemic model} is a triple \(M = (W, \sim, V)\), where
\begin{itemize}
    \item \(W\) is a non-empty set of states;
    \item \(\sim:A \rightarrow \mathcal{P}(W \times W)\) assigns an equivalence relation to each agent; we will denote relation assigned to agent $a \in A$ by $\sim_a$;
    \item \(V:P \rightarrow \mathcal{P}(W)\) assigns a set of states to each propositional variable. 
\end{itemize}
A pair \((W,\sim)\) is called an \emph{epistemic frame}, and a pair \((M,w)\) with \(w \in W\) is called a \emph{pointed model}. An announcement in a pointed model \((M,w)\) results in an \textit{updated pointed model}  
\((M^\varphi, w)\).
Here \(M^\varphi = (W^\varphi, \sim^\varphi, V^\varphi)\), and \(W^\varphi = \llbracket \varphi \rrbracket_M\), \(\sim^\varphi_a = \sim_a \cap\) \((\llbracket \varphi \rrbracket_M \times \llbracket \varphi \rrbracket_M)\), and \(V^\varphi (p) = V(p) \cap \llbracket \varphi \rrbracket_M\). Generally speaking, an updated pointed model \((M^\varphi,w)\) is a restriction of the original one to the states where \(\varphi\) holds. 
\end{definition}

Let \(\mathcal{L}_{EL}^G\) denote the set of formulas of the type \(\bigwedge_{i \in G} K_i \varphi_i\), where for every \(i \in G\) it holds that \(\varphi_i \in \mathcal{L}_{EL}\). These are the formulas we will be quantifying
over in modalities of the form $[G]$ and $[ \! \langle G \rangle \! ]$. 
%In other words, formulae of \(\mathcal{L}_{EL}^{G}\) are of the type `agent \(i\) from coalition \(G\) knows \(\varphi_i\).'

\begin{definition} (Semantics)
Let a pointed model \((M,w)\) with \(M = (W\), \(\sim, V)\), \(a \in A\), and \(\varphi\), \(\psi \in \mathcal{L}_{CoGAL}\) be given.
\begin{center}
\(
\begin{array} {lcl}
(M,w) \models p  &\textrm{iff} &w \in V(p)\\

(M,w) \models \neg \varphi &\textrm{iff} &(M,w) \not \models \varphi\\

(M,w) \models \varphi \wedge \psi &\textrm{iff} &(M,w) \models \varphi \textrm{ and } (M,w) \models \psi\\

(M,w) \models K_a\varphi &\textrm{iff} &\forall v \in W: w \sim_a v \textrm{ implies } (M,v) \models \varphi\\

(M,w) \models [\varphi]\psi &\textrm{iff} &(M,w) \models \varphi \textrm{ implies } (M^\varphi,w) \models \psi\\ 

(M,w) \models [G] \varphi &\textrm{iff} &\forall \psi {\in} \mathcal{L}_{EL}^G: (M,w) \models  [ \psi ] \varphi\\

(M,w) \models [ \! \langle G \rangle \! ]\varphi &\textrm{iff} &\forall \psi {\in}  \mathcal{L}_{EL}^G \ \exists \chi {\in}  \mathcal{L}_{EL}^{A \setminus G}:  (M,w) \models \psi \rightarrow \langle \psi \wedge \chi \rangle \varphi\\

\end{array}\)
\end{center}

\end{definition}
\noindent
 Formula \(\varphi\) is called \textit{valid} if for any pointed model \((M,w)\) it holds that \((M,w) \models \varphi\).  

The semantics for the `diamond' versions of knowledge, public and group announcement operators (\(\widehat {K}_a \varphi\), \(\langle \varphi \rangle \psi\), and \(\langle G \rangle \varphi\) respectively) are obtained by changing \(\forall\) to \(\exists\) and `implies' to `and' in the corresponding lines. The semantics for a dual of the coalition announcement operator is as follows:

\begin{center}
\(
\begin{array}{lcl}

(M,w) \models \langle \! [ G ] \! \rangle \varphi &\textrm{iff} &\exists \psi {\in}  \mathcal{L}_{EL}^G \ \forall \chi {\in}  \mathcal{L}_{EL}^{A\setminus G}: (M,w) \models \psi\wedge [ \psi \wedge \chi ] \varphi,\\

\end{array}\)
\end{center}
which corresponds to `there is an announcement by agents from \(G\), such that whatever other agents \(A \setminus G\) announce at the same time, \(\varphi\) holds.' 

Note that following \cite{balbiani07, balbiani08, agotnes10, agotnes08, agotnes16, balbiani15, vanditmarsch12, agotnes14} we restrict formulas which agents in a group or coalition can announce to formulas of \( \mathcal{L}_{EL}\).
This allows us to avoid circularity in the definition.

\subsection{Axiomatisation and Some Logical Properties}
In this section we present an axiomatisation of \(\mathbf{CoGAL}\) and show its soundness.

\begin{definition}
\label{def::axiomatisation}
\emph{Axiomatisation} of \(\mathbf{CoGAL}\) is a union of axiomatisation of \(\mathbf{GAL}\) \cite{agotnes10}, interaction axiom for group and coalition announcements A11, rule of inference for coalition announcements R6, and necessitation R4.

\(  \begin{array}[t]{ll}
    (A0) &\textrm{instantiations of propositional tautologies},  \\

    (A1) &K_a(\varphi \rightarrow \psi) \rightarrow (K_a \varphi \rightarrow K_a \psi), \\      

    (A2) &K_a\varphi \rightarrow \varphi, \\

        (A3) &K_a\varphi \rightarrow K_a K_a \varphi,  \\  

        (A4) &\neg K_a \varphi \rightarrow K_a \neg K_a \varphi,  \\

        (A5) &[\varphi]p \leftrightarrow (\varphi \rightarrow p),  \\

        (A6) &[\varphi] \neg \psi \leftrightarrow (\varphi \rightarrow \neg [\varphi] \psi),  \\

        (A7) &[\varphi](\psi \wedge \chi) \leftrightarrow ([\varphi]\psi \wedge [\varphi]\chi),\\

        (A8) &[\varphi]K_a\psi \leftrightarrow (\varphi \rightarrow K_a[\varphi]\psi), \\

        (A9) &[\varphi][\psi]\chi \leftrightarrow [\varphi \wedge [\varphi]\psi] \chi, \\       
    
        %(A12) &\neg \langle \! [ G ] \! \rangle \bot, \\
\end{array}\)
\(
\begin{array}[t]{ll}
	%(A13) &\langle \! [ G ] \! \rangle \top, \\

	%(A14) &\neg \langle \! [ \emptyset ] \! \rangle \neg \varphi \rightarrow \langle \! [ A ] \! \rangle \varphi, \\

	%(A15) &\langle \! [ G ] \! \rangle (\varphi \wedge \psi) \rightarrow \langle \! [ G ] \! \rangle \varphi, \\      

	%(A16) &\langle \! [ G ] \! \rangle \varphi \wedge\langle \! [ H ] \! \rangle \psi \rightarrow \langle \! [ G \cup H ] \! \rangle (\varphi \wedge \psi),		\textrm{ if } G \cap H = \emptyset, \\

  	(A10) &[G]\varphi \rightarrow [\psi]\varphi, \textrm{ where } \psi \in  \mathcal{L}_{EL}^G,\\

 	(A11) &\langle \! [ G ] \! \rangle \varphi \rightarrow \langle G \rangle [ A \setminus G ] \varphi, \\

 	(R0) &\vdash \varphi, \varphi \rightarrow \psi \Rightarrow \, \vdash \psi, \\

	(R1) &\vdash \varphi \Rightarrow \, \vdash K_a \varphi, \\

	(R2) &\vdash \varphi \Rightarrow \, \vdash [\psi] \varphi, \\

	(R3) &\vdash \varphi \Rightarrow \, \vdash [G] \varphi, \\

	(R4) &\vdash \varphi \Rightarrow \, \vdash [ \! \langle G \rangle \! ] \varphi, \\

	(R5) & (\forall \psi {\in}  \mathcal{L}_{EL}^G \vdash \eta ([\psi]\varphi)) \Rightarrow \, \vdash \eta ([G]\varphi), \\

 	(R6) & (\forall \psi {\in}  \mathcal{L}_{EL}^G \ \exists \chi {\in}  \mathcal{L}_{EL}^{A \setminus G}\\
       	& \vdash \eta (\psi \rightarrow \langle \psi \wedge \chi \rangle \varphi)) \Rightarrow  \vdash \eta ([ \! \langle G \rangle \! ]\varphi).\\

	%(R7) &\vdash \varphi \leftrightarrow \psi \Rightarrow \, \vdash \langle \! [ G ] \! \rangle \varphi \leftrightarrow \langle \! [ G ] \! \rangle \psi.\\

\end{array}\)

\end{definition}

So, \(\mathbf{CoGAL}\) is the smallest subset of \(\mathcal{L}_{CoGAL}\) that contains all the axioms \(A0\) -- \(A11\) and closed under rules of inference \(R0\) -- \(R6\). Elements of \(\mathbf{CoGAL}\) are called \textit{theorems}. Note that \(R5\) and \(R6\) are infinitary rules: they require an infinite number of premises. Finding finite axiomatisations of any of \(\mathbf{APAL}\), \(\mathbf{GAL}\), or \(\mathbf{CAL}\) is an open problem. Note also that \(\mathbf{CoGAL}\) includes coalition logic \cite{pauly02}, that is all the axioms of the latter are validities of \(\mathbf{CoGAL}\) and a rule of inference preserves validity (see Appendix A).

\begin{definition} (Soundness and completeness)
	An axiomatisation is \textit{sound}, if for any formula \(\varphi\) of the language, it holds that \(\varphi \in \mathbf{CoGAL}\) implies \(\varphi\) is valid. And vice versa for \textit{completeness}.
\end{definition}

Soundness of \(A0\)--\(A4\), \(R0\), and \(R1\) is due to soundness of \(\mathbf{S5}\). Axioms \(A5\)--\(A9\) and rule of inference \(R3\) are sound, since \(\mathbf{PAL}\) is sound \cite{del}. Soundness of axiom \(A10\) and rules of inference \(R3\) and \(R5\) was shown in \cite{agotnes10}. %Axioms and rules of inference of coalition logic (\(A12\)--\(A15\), \(R7\)) are sound as well \cite{agotnes08}.  
We show soundness of \(R4\), \(R6\) in Proposition \ref{prop::soundness}, and validity of \(A11\) in Proposition \ref{prop::axiom}.

\begin{proposition}
\label{prop::soundness}
    \(R4\) and \(R6\) are sound, that is, they preserve validity.
\end{proposition}
\begin{proof}
A proof is given in Appendix B (Proposition \ref{prop::app1}).
 \end{proof}

Validity of \(A11\) was stated to be an open question in \cite{vanditmarsch12, agotnes16}. Informally, the idea of our proof is as follows. Let us examine the axiom: \(\langle \! [ G ] \! \rangle \varphi \rightarrow \langle G \rangle [ A \setminus G ] \varphi\). In the antecedent, all the agents make announcements simultaneously. In the consequent, the agents in \(A \setminus G\) know the announcement $\psi$ made by the agents in \(G\). 
In the updated model \((M^\psi,w)\) the agents in \(A \setminus G\) may learn some \emph{new} epistemic formulae \(\chi\) which they did not know before the announcement. We need to make sure that these new formulae cannot allow them to make \(\varphi\) false. However, since \(\psi\) is true in the initial model, and \(\chi\) in the updated one, agents in \(A \setminus G\) can always make an announcement that they know that after the announcement of \(\psi\), \(\chi\) holds.
This announcement, made simultaneously with the announcement by \(G\), `models' the effect of announcing $\chi$ later. Returning to our example (Figure \ref{fig::tark1}), whichever formulae \(\psi_1\) and \(\psi_2\) Ann and Bob announce, and whichever formula \(\varphi\) Cath learns afterwards, she can always announce \([\psi_1 \wedge \psi_2] K_c \varphi\) simultaneously with them in the initial situation. Informally, if after Bob's announcement of $\neg p$, 
Cath learns that $\neg p$, she can announce: `If you say that \(\neg p\) holds, then I will know it,' or $[\neg p] K_c \neg p$. We use this idea to prove that if the agents in $A \setminus G$ can prevent $\varphi$ after the announcement by $G$, then they could have prevented it before. 

Due to restriction of announcements to formulas of epistemic logic, we cannot directly employ public announcement operators in agents' `utterances.' In order to avoid this, we use the standard translation of $\mathbf{PAL}$ into epistemic logic.
\begin{definition}
\emph{Translation function} \(t:\mathcal{L}_{PAL} \rightarrow \mathcal{L}_{EL}\) \cite{del} is defined as follows:
\begin{center}
\(
\begin{array}[t] {lcl}
t(p) &= &p, \\

t(\neg \varphi) &= &\neg t(\varphi), \\

t(\varphi \wedge \psi) &= &t(\varphi) \wedge t(\psi), \\

t(K_a \varphi) &= &K_at(\varphi), \\

t([\varphi]p) &= &t(\varphi \rightarrow p),\\

\end{array}\)
\( \begin{array}[t]{lcl}

t([\varphi] \neg \psi) &= &t(\varphi \rightarrow \neg [\varphi]\psi),\\

t([\varphi](\psi \wedge \chi)) &= &t([\varphi] \psi \wedge [\varphi] \chi),\\

t([\varphi]K_a \psi) &= &t(\varphi \rightarrow K_a [\varphi]\psi),\\

t([\varphi][\psi]\chi) &= &t([\varphi \wedge [\varphi]\psi]\chi).\\
\end {array} \)
\end{center}

Every \(\varphi \in \mathcal{L}_{PAL}\) is equivalent to \(t(\varphi) \in \mathcal{L}_{EL}\).
\end{definition}

Now we show that for every announcement of agents' knowledge in some updated pointed model \((M^\psi,w)\) there is an equivalent announcement in the original one (i.e. in \((M,w)\)).

\begin{lemma}
\label{lemma::1}
    Let \(a, \mathellipsis, b \in A\). The following formula is valid for all \(\psi\), \(\chi_a, \mathellipsis,\) \(\chi_b \in \mathcal{L}_{EL}\):
$$[\psi \wedge K_a t([\psi]\chi_a) \wedge \mathellipsis \wedge K_b t([\psi]\chi_b)] \varphi \leftrightarrow [\psi][K_a \chi_a \wedge \mathellipsis \wedge K_b \chi_b]\varphi$$
\end{lemma}

\begin{proof}
    Suppose that for some pointed model \((M,w)\) it holds that \((M,w) \models [\psi \wedge K_a t([\psi]\chi_a) \wedge \mathellipsis \wedge K_b t([\psi]\chi_b)] \varphi\). 
    By propositional reasoning, it is equivalent to \((M,w) \models [\psi \wedge (\psi \rightarrow K_a t([\psi]\chi_a)) \wedge \mathellipsis \wedge (\psi \rightarrow K_b t([\psi]\chi_b))] \varphi\), and, by equivalence of a formula and its translation, the latter is equivalent to \((M,w) \models [\psi \wedge (\psi \rightarrow K_a [\psi]\chi_a) \wedge \mathellipsis \wedge (\psi \rightarrow K_b [\psi]\chi_b)] \varphi\). 
    By \(A8\), we have that \((M,w) \models [\psi \wedge [\psi] K_a\chi_a \wedge \mathellipsis \wedge [\psi] K_b\chi_b] \varphi\), and, by \(A7\), \((M,w) \models [\psi \wedge [\psi] (K_a\chi_a \wedge \mathellipsis \wedge K_b\chi_b)] \varphi\).
    Finally, by \(A9\), the latter is equivalent to \((M,w) \models [\psi] [K_a\chi_a \wedge \mathellipsis \wedge K_b \chi_b] \varphi\). 
\end{proof}

We use Lemma \ref{lemma::1} to show validity of axiom \(A11\).

\begin{proposition}
\label{prop::axiom}
    \(\langle \! [ G ] \! \rangle  \varphi \rightarrow \langle G \rangle [ A \setminus G ] \varphi\) is valid.
\end{proposition}

\begin{proof}
    Suppose to the contrary that for some pointed model \((M,w)\) it holds that \((M,w) \models \langle \! [ G ] \! \rangle \varphi\) and \((M,w) \not \models \langle G \rangle [ A \setminus G ] \varphi\). 
    From \((M,w) \models \langle \! [ G ] \! \rangle \varphi\), by the semantics, 
$$\exists \psi {\in} \mathcal{L}_{EL}^G   \   \forall \chi_a, \mathellipsis \chi_b {\in} \mathcal{L}_{EL}:  
(M,w) \models \psi \wedge [ \psi \wedge K_a \chi_a \wedge \mathellipsis \wedge K_b \chi_b ] \varphi.$$ 
Let us call \(\psi_G\) the formula that $G$ can announce to enforce $\varphi$.
From \((M,w) \not \models \langle G \rangle [ A \setminus G ] \varphi\),
$$\forall \psi^\prime {\in} \mathcal{L}_{EL}^G   \   \exists \chi_a^\prime, \mathellipsis, \chi_b^\prime {\in} \mathcal{L}_{EL}: (M,w) \not \models \langle \psi^\prime \rangle [ K_a \chi_a^\prime \wedge \mathellipsis \wedge K_b \chi_b^\prime ] \varphi.$$
In particular, for $\psi' = \psi_G$,
$$\exists \chi_a^\prime, \mathellipsis, \chi_b^\prime {\in} \mathcal{L}_{EL}: (M,w) \not \models \langle \psi_G \rangle [ K_a \chi_a^\prime \wedge \mathellipsis \wedge K_b \chi_b^\prime ] \varphi.$$
Since $\psi_G$ is true in $(M,w)$, this is equivalent to
$$\exists \chi_a^\prime, \mathellipsis, \chi_b^\prime {\in} \mathcal{L}_{EL}: (M,w) \not \models [\psi_G] [ K_a \chi_a^\prime \wedge \mathellipsis \wedge K_b \chi_b^\prime ] \varphi.$$
    By Lemma \ref{lemma::1}, the latter is equivalent to 
$$\exists \chi_a^\prime, \mathellipsis, \chi_b^\prime \in \mathcal{L}_{EL}  : (M,w) \not \models [ \psi_G \wedge K_a t([\psi_G]\chi_a^\prime) \wedge \mathellipsis \wedge K_b t([\psi_G]\chi_b^\prime) ] \varphi.$$
Since  $t([\psi_G]\chi_a^\prime), \mathellipsis, t([\psi_G]\chi_b^\prime)$ are in $\mathcal{L}_{EL}$,   
we have the contradiction with
\[\forall \chi_a, \mathellipsis \chi_b {\in} \mathcal{L}_{EL}:
(M,w) \models \psi_G \wedge [\psi_G \wedge K_a \chi_a \wedge \mathellipsis \wedge K_b \chi_b ] \varphi.\qedhere\] 
\end{proof}
Proposition \ref{prop::axiom} allows us to prove Lindenbaum Lemma (Proposition \ref{prop::lindenbaum}) for \(\mathbf{CoGAL}\). But before that, let us show some properties of the logic. The following validity shows that if some formula \(\varphi\) can be achieved by two coalition announcements, it can be achieved by a single joint coalition announcement as well. The validity was known only for the case of group announcements in \(\mathbf{GAL}\) \cite{agotnes10}. We show that this also holds for coalition announcements.

\begin{proposition}
\label{prop::long}
    \(\langle \! [ G ] \! \rangle \langle \! [ H ] \! \rangle \varphi  \rightarrow \langle \! [ G \cup H ] \! \rangle \varphi\) is valid.
\end{proposition}

\begin{proof}
The proof is presented in Appendix B (Proposition \ref{prop::app2}).
\end{proof}

\begin{corollary}
     \(\langle \! [ G ] \! \rangle \langle \! [ G ] \! \rangle \varphi  \rightarrow \langle \! [ G ] \! \rangle \varphi\) is valid.    
\end{corollary}

The other direction of Proposition \ref{prop::long} does not hold.  Whether  \(\langle \! [ G \cup H ] \! \rangle \varphi \rightarrow \langle \! [ G ] \! \rangle \langle \! [ H ] \! \rangle \varphi\) is valid was posed as an open question in \cite{agotnes16}. We settle this question by presenting a counterexample.

\begin{proposition}
    \(\langle \! [ G \cup H ] \! \rangle \varphi \rightarrow \langle \! [ G ] \! \rangle \langle \! [ H ] \! \rangle \varphi\) is not valid.
\end{proposition}

\begin{proof}
    Let \(G = \{a\}, H = \{b\}\), and \(\varphi:= K_b(p \wedge q \wedge r) \wedge \neg K_a(p \wedge q \wedge r) \wedge \neg K_c(p \wedge q \wedge r)\). Informally, \(\varphi\) says that agent \(b\) knows that the given propositional variables are true, and agents \(a\) and \(c\) do not. Consider the model \(M\) in Figure \ref{fig:counter} (reflexive arrows are omitted for convenience).
    By the semantics, \((M,\blacksquare) \models \langle \! [ \{ a, b \} ] \! \rangle \varphi\) iff \(\exists (K_a\psi \wedge K_b\chi)\), \(\psi, \chi \in \mathcal{L}_{EL}\), \(\forall K_c \tau\), \(\tau \in \mathcal{L}_{EL}: (M,\blacksquare) \models K_a\psi \wedge K_b\chi \wedge [ K_a\psi \wedge K_b\chi \wedge K_c \tau ] \varphi\). Let \(\psi\) be \(q\), and \(\chi\) be \(\top\). 
    Observe that \((M,\blacksquare) \models K_a q \wedge K_b \top\).
    Moreover, \(c\) does not know any formula that she can announce to avoid \(\varphi\). An informal argument is as follows. Whatever \(c\) announces in this situation, she cannot avoid \(b\) learning \(p \wedge q \wedge r\). In order to make \(a\) learn that \(p \wedge q \wedge r\), \(c\) has to announce something of the form \(\psi \rightarrow p\), since she does not know the value of \(p\) herself.
    Formula \(\psi\) can be neither \(r\) nor \(q\), because \(c\) does not know their truth values. Also, it cannot be a statement of \(b\)'s knowledge, since in every \(q\)-world accessible by \(c\), \(b\)'s knowledge is only a reflexive arrow. 
    It cannot be \(a\)'s or \(c\)'s knowledge either, since in this case \(a\) would have known \(p\) herself, and \(c\)'s relation between \(q\)-states is universal.
    
    In the consequent, we have \((M, \blacksquare) \models \langle \! [ a ] \! \rangle \langle \! [ b ] \! \rangle \varphi\). By the semantics, \(\exists \psi {\in} \mathcal{L}_{EL}^a\) \ \(\forall \chi {\in} \mathcal{L}_{EL}^{b\cup c}\): \((M, \blacksquare) \models \psi \wedge [ \psi \wedge \chi ] \langle \! [ b ] \! \rangle \varphi\).
   Let us fix such a \(\psi\), and let \(\chi := K_b p \wedge K_c \top\). Then \((M, \blacksquare) \models \psi \wedge [ \psi \wedge K_b p \wedge K_c \top ] \langle \! [ b ] \! \rangle \varphi\).
    Observe that no matter what \(a\) announces, \(K_b p\) `forces' \(a\) to learn that \(p \wedge q \wedge r\), and whatever is announced in the updated model \((M^{\psi \wedge K_b p \wedge K_c \top }, \blacksquare)\), \(a\)'s knowledge of \(p \wedge q \wedge r\) and, hence, falsity of \(\varphi\) remains. So, \((M, \blacksquare) \not \models  \langle \! [ a ] \! \rangle \langle \! [ b ] \! \rangle \varphi\).
\end{proof}

\begin{figure}[h]
    \centering
    \includegraphics {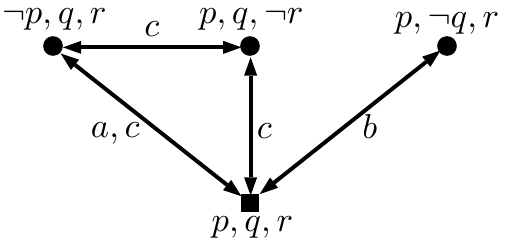}
    \caption{A counterexample}
    \label{fig:counter}
\end{figure}

\subsection{Completeness}
In order to prove completeness of \(\mathbf{CoGAL}\), we expand and modify the completeness proof for \(\mathbf{APAL}\) \cite{balbiani08, balbiani15, balbiani15a}. Although the proof is partially based upon the classic canonical model approach, we have to ensure that construction of maximal consistent theories (Proposition \ref{prop::lindenbaum}) allows us to include infinite amount of formulas for cases of coalition announcements. This corresponds to `whatever other agents may say.' But before we start, let us state two auxiliary lemmas.

\begin{lemma}
\label{lemma::2}
	Let \(\varphi, \psi \in \mathcal{L}_{CoGAL}\). If \(\varphi \rightarrow \psi\) is a theorem, then \(\eta (\varphi) \rightarrow \eta (\psi)\) is a theorem as well.
\end{lemma}

\begin{proof}
    A simple induction on the complexity of \(\eta\).
\end{proof}

\begin{lemma}
\label{lemma::3}
    \([\chi] \langle \psi \rangle \varphi\) \(\rightarrow\) \([\chi] \langle G \rangle \varphi\), where \(\chi \in \mathcal{L}_{CoGAL}\), and \(\psi \in \mathcal{L}_{EL}^G\), is a theorem.
\end{lemma}

\begin{proof}
    Let \(\psi \in \mathcal{L}_{EL}^G\). By \(A10\), \([G] \neg \varphi \rightarrow [\psi] \neg \varphi\) is a theorem.
    By contraposition, we have that \(\langle \psi \rangle \varphi\) \(\rightarrow\) \(\langle G \rangle \varphi\) is also a theorem. By \(R2\) and distribution over implication, we infer  \([\chi] \langle \psi \rangle \varphi\) \(\rightarrow\) \([\chi] \langle G \rangle \varphi\). 
\end{proof}
Now, the first part of the proof up to Proposition \ref{prop::lindenbaum} is based on \cite{balbiani08}.
\begin{definition}
    A set of formulae \(x\) is called a \emph{theory} if and only if it contains \(\mathbf{CoGAL}\), and is closed under \(R0, R5,\) and \(R6\). A theory \(x\) is consistent if and only if \(\bot \not \in x\), and is maximal if and only if for all \(\varphi \in \mathcal{L}_{CoGAL}\) it holds that either \(\varphi \in x\) or \(\neg \varphi \in x\). 
\end{definition}

\begin{proposition}
\label{prop::mct}
    Let \(x\) be a theory, \(\varphi, \psi \in \mathcal{L}_{CoGAL}\), and \(a \in A\). The following are theories: \(x + \varphi = \{\psi: \varphi \rightarrow \psi \in x\}, K_a x = \{\varphi: K_a \varphi \in x\}\), and \([\varphi]x = \{\psi: [\varphi]\psi \in x\}\).
\end{proposition}

\begin{proof}
The proof is an extension of the one from \cite{balbiani08} (see Appendix B, Propostion \ref{prop::app3}). \qedhere
\end{proof}

\begin{proposition}
\label{prop::consistency}
 %Let \(x\) be a theory, and \(\varphi \in \mathcal{L}_{CoGAL}\). Then, \(x + \varphi\) is consistent iff \(\neg \varphi \not \in x\).
Let \(\varphi \in \mathcal{L}_{CoGAL}\). Then, \(\mathbf{CoGAL} + \varphi\) is consistent iff \(\neg \varphi \not \in \mathbf{CoGAL}\).
\end{proposition}

\begin{proof}
The proof is given in Appendix B (Proposition \ref{prop::balbiani}). \qedhere
\end{proof} 

The following proposition is a variation of Lindenbaum Lemma. Validity of axiom \(A11\) allows us to expand the corresponding proof for \(\mathbf{APAL}\), and to deal with having two different quantifiers at the same time.

\begin{proposition}
\label{prop::lindenbaum}
    Every consistent theory \(x\) can be extended to a maximal consistent theory \(y\).
\end{proposition}

\begin{proof}
    Let \(\psi_0, \psi_1, \mathellipsis\) be an enumeration of formulae of the language, and let \(y_0 = x\). 
    Suppose that for some \(n \geq 0\), \(y_n\) is a consistent theory, and \(x \subseteq y_n\). 
    If \(y_n + \psi_n\) is consistent, then \(y_{n+1} = y_n + \psi_n\). 
    Otherwise, if \(\psi_n\) is not a conclusion of either \(R5\) or \(R6\), \(y_{n+1} = y\). 
    If \(\psi_n\) is a conclusion of \(R5\), we enumerate all the subformulae of \(\psi_n\) which contain group announcement modalities \([G]\). 
    Let \(\eta_1 ([G]\varphi_1), \mathellipsis, \eta_k ([G]\varphi_k)\) be all these subformulae. 
    Then \(y_n^0, \mathellipsis, y_n^k\) is a sequence of consistent theories, where \(y_n^0 = y_n\), and for some \(i < k\), \(y_n^i\) is a consistent theory containing \(y_n\) and \(\neg \eta_i ([G]\varphi_i)\). 
    Since \(y_n^i\) is closed under \(R5\), there exists \(\chi \in \mathcal{L}_{EL}^G\) such that \(\eta_i ([\chi]\varphi_i) \not \in y_n^i\). 
    Hence, \(y_n^{i+1} = y_n^i + \neg \eta_i ([\chi]\varphi_i)\), and \(y_{n+1} = y_n^k\).
    
    Now, we consider the case when \(\psi_n\) is a conclusion of \(R6\). We enumerate all the subformulae of \(\psi_n\) which contain coalition announcement modalities \([ \! \langle G \rangle \! ]\). 
    Let \(\eta_1 ([ \! \langle G \rangle \! ]\varphi_1)\), \(\mathellipsis\), \(\eta_k ([ \! \langle G \rangle \! ]\varphi_k)\) be all these subformulae. 
    Then \(y_n^0, \mathellipsis, y_n^k\) is a sequence of consistent theories, where \(y_n^0 = y_n\), and for some \(i < k\), \(y_n^i\) is a consistent theory containing \(y_n\) and \(\neg \eta_i ([ \! \langle G \rangle \! ]\varphi_i)\).
    By \(A11\), this means that \(\neg \eta_i ([G] \langle A \setminus G \rangle \varphi_i) \in y_n^i\).
    Since \(y_n^i\) is closed under \(R5\), there exists \(\chi \in \mathcal{L}_{EL}^G\) such that \(\eta_i ([\chi] \langle A \setminus G \rangle \varphi_i)\) \(\not \in y_n^i\).   
    Hence, \(y_n^{i+1} = y_n^i + \neg \eta_i ([\chi] \langle A \setminus G \rangle \varphi_i)\), and \(y_{n+1} = y_n^k\).
    Note that since for all \(\tau \in \mathcal{L}_{EL}^{A \setminus G}\) \(\eta ([\chi] \langle \tau \rangle \varphi) \rightarrow \eta ([\chi] \langle A \setminus G \rangle \varphi)\) are theorems (by Lemmas \ref{lemma::2} and \ref{lemma::3}), they and their contrapositions are already in \(y_n^i\) (since \(y_n^i\) is a theory). Thus, adding \(\neg \eta_i ([\chi] \langle A \setminus G \rangle \varphi_i)\) to \(y^i_n\) adds all the \(\neg \eta_i ([\chi] \langle \tau \rangle \varphi_i)\) for \(\tau \in \mathcal{L}_{EL}^{A \setminus G}\) as well.
    
    Finally, \(y\) is a maximal consistent theory, and \(x \subseteq y\).
\end{proof} 

The rest of the proof is an expansion of the one from \cite{balbiani15}. It employs induction on complexity of formulae to prove Truth Lemma (Proposition \ref{prop::truth}) and, ultimately, completeness (Proposition \ref{prop::completeness}) of \(\mathbf{CoGAL}\).

\begin{definition}
    The \emph{size} of some formula \(\varphi \in \mathcal{L}_{CoGAL}\) is defined as follows:
    \begin{enumerate}
        \item \(Size (p) = 1\),
        \item \(Size (\neg \varphi)\) \(=\) \(Size (K_a \varphi) = Size ([G] \varphi) =\) \(Size ([ \! \langle G \rangle \! ] \varphi) =\) \(Size (\varphi) + 1\), 
        \item \(Size (\varphi \wedge \psi) = Size (\varphi) + Size (\psi) + 1\),
        \item \(Size ([\psi] \varphi) = Size (\psi) + 3 \cdot Size (\varphi)\).
    \end{enumerate}
    
    The \([]\)-\emph{depth} is defined as follows:
    \begin{enumerate}
        \item \(d_{[]} (p) = 0\),
        \item \(d_{[]} (\neg \varphi) = d_{[]} (K_a \varphi) = d_{[]} ([ \! \langle G \rangle \! ]) = d_{[]} (\varphi)\),
        \item \(d_{[]} (\varphi \wedge \psi) = \mathrm{max} \{d_{[]} (\varphi), d_{[]} (\psi)\}\),
        \item \(d_{[]} ([\psi] \varphi) = d_{[]} (\psi) + d_{[]} (\varphi)\),
        \item \(d_{[]} ([G]\varphi) = d_{[]} (\varphi) + 1\).
    \end{enumerate}
    
    The \([ \! \langle \! \rangle \! ]\)-\emph{depth} is the same as \([]\), with the following exceptions.  
    \begin{enumerate}
    \item \(d_{[ \! \langle \! \rangle \! ]} ([G]\varphi) = d_{[ \! \langle \! \rangle \! ]} (\varphi)\),
    \item \(d_{[ \! \langle \! \rangle \! ]} ([ \! \langle G \rangle \! ]\varphi) = d_{[ \! \langle \! \rangle \! ]} (\varphi) + 1\).
    \end{enumerate}
\end{definition}

\begin{definition}
	The binary relation \(<^{Size}_{[], [ \! \langle \! \rangle \! ]}\) between \(\varphi, \psi \in \mathcal{L}_{CoGAL}\) is defined as follows: \(\varphi <^{Size}_{[], [ \! \langle \! \rangle \! ]} \psi\) iff \(d_{[ \! \langle \! \rangle \! ]} (\varphi) < d_{[ \! \langle \! \rangle \! ]} (\psi)\), or, otherwise,  \(d_{[ \! \langle \! \rangle \! ]} (\varphi) = d_{[ \! \langle \! \rangle \! ]} (\psi)\), and either \(d_{[]} (\varphi) < d_{[]} (\psi)\), or \(d_{[]} (\varphi) = d_{[]} (\psi)\) and \(Size (\varphi) < Size (\psi)\). The relation is a well-founded strict partial order between formulae.
\end{definition}
  
Now, we ensure that the order of complexity is preserved. 
Case \([\bigwedge_{i \in G} K_i \psi_i]\varphi <^{Size}_{[], [ \! \langle \! \rangle \! ]} [G]\varphi\) is obvious, since the public announcement on the left-hand side of the inequality is epistemic, and for any epistemic formula \(\psi\), \(d_{[]}(\psi) = d_{[ \! \langle \! \rangle \! ]} (\psi) = 0\). 
Case \([\chi] [\bigwedge_{i \in G} K_i \psi_i] \varphi <^{Size}_{[], [ \! \langle \! \rangle \! ]} [\chi][G]\varphi\) holds for the same reason. 
The cases for coalitions are identical: \(\bigwedge_{i \in G} K_i \psi_i \rightarrow \langle \bigwedge_{i \in G} K_i \psi_i \wedge \bigwedge_{j \in A \setminus G} K_j \chi_j \rangle \varphi <^{Size}_{[], [ \! \langle \! \rangle \! ]} [ \! \langle G \rangle \! ]\varphi\), and also \([\tau] (\bigwedge_{i \in G} K_i \psi_i \rightarrow \langle \bigwedge_{i \in G} K_i \psi_i\) \(\wedge\) \(\bigwedge_{j \in A \setminus G} K_j \chi_j \rangle \varphi)\) \(<^{Size}_{[], [ \! \langle \! \rangle \! ]}\) \([\tau][ \! \langle G \rangle \! ]\varphi\).

\begin{definition}
    The \emph{canonical model} is the model \(M^C = (W^C, \sim^C, V^C)\), where
    \begin{itemize}
        \item \(W^C\) is the set of all maximal consistent theories,
        \item \(\sim^C\) is defined as \(x \sim^C_a y\) iff \(K_a x \subseteq y\),
        \item \(x \in V^C\) iff \(p \in x\).
    \end{itemize}
    Relation \(\sim^C\) is equivalence due to axioms \(A2\), \(A3\), and \(A4\).
\end{definition}

\begin{definition}
    Let \(\varphi \in \mathcal{L}_{CoGAL}\). 
    Condition \(P(\varphi)\): for all maximal consistent theories \(x\), \(\varphi \in x\) iff \(x \in \llbracket \varphi \rrbracket_{M_C}\). 
    Condition \(H(\varphi)\): for all \(\psi \in\mathcal{L}_{CoGAL}\), if \(\psi <^{Size}_{[], [ \! \langle \! \rangle \! ]} \varphi\), then \(P(\psi)\). 
\end{definition}

\begin{proposition}
\label{prop::ph}
    For all \(\psi \in \mathcal{L}_{CoGAL}\), if \(H(\varphi)\), then \(P(\varphi)\).
\end{proposition}

\begin{proof}
    Suppose \(H(\varphi)\) holds, and let \(x\) be a maximal consistent theory. The proof is by induction on \(<^{Size}_{[], [ \! \langle \! \rangle \! ]}\)-complexity of formulae. 
    Most of the cases were proved in \cite{balbiani15}.
    We prove here only remaining instances involving group and coalition announcements.
    
    \emph{Case} \(\varphi = [G] \psi\). Suppose that \([G] \psi \in x\). Since \(x\) is closed under \(R5\), this is equivalent to \(\forall \chi \in\mathcal{L}_{EL}^G\): \([\chi] \psi \in x\). 
    By the fact that \([\chi] \psi <^{Size}_{[], [ \! \langle \! \rangle \! ]} [G] \psi\), the latter holds if and only if \(x \in \llbracket [\chi] \psi \rrbracket_{M_C}\) for all \(\chi \in \mathcal{L}_{EL}^G\), which is equivalent to \(x \in \llbracket [G] \psi \rrbracket_{M_C}\) by the semantics.
    
    \emph{Case} \(\varphi = [\chi] [G] \psi\). Suppose that \([\chi] [G] \psi \in x\). Since \([\chi] [G] \psi\) is a necessity form and \(x\) is closed under \(R5\), this is equivalent to \(\forall \tau \in \mathcal{L}_{EL}^G\): \([\chi] [\tau] \psi \in x\). 
    By the fact that \([\chi] [\tau] \psi <^{Size}_{[], [ \! \langle \! \rangle \! ]} [\chi] [G] \psi\), the latter holds if and only if \(x \in \llbracket [\chi] [\tau] \psi \rrbracket_{M_C}\) for all \(\tau \in \mathcal{L}_{EL}^G\), which is equivalent to \(x \in \llbracket [\chi] [G] \psi \rrbracket_{M_C}\) by the semantics. 
    
    \emph{Case} \(\varphi = [ \! \langle G \rangle \! ] \psi\). Suppose that \([ \! \langle G \rangle \! ] \psi \in x\). Since \(x\) is closed under \(R6\), this is equivalent to \(\forall \chi {\in} \mathcal{L}_{EL}^G\) \ \(\exists \tau {\in} \mathcal{L}_{EL}^{A \setminus G}\): \( \chi \rightarrow \langle \chi \wedge \tau \rangle \psi \in x\). 
    By the fact that \(\chi \rightarrow \langle \chi \wedge \tau \rangle \psi <^{Size}_{[], [ \! \langle \! \rangle \! ]} [ \! \langle G \rangle \! ] \psi\), the latter holds if and only if \(\forall \chi {\in} \mathcal{L}_{EL}^G\) \ \(\exists \tau {\in} \mathcal{L}_{EL}^{A \setminus G}\): \(x \in \llbracket \chi \rightarrow \langle \chi \wedge \tau \rangle \psi \rrbracket_{M_C}\)which is equivalent to \(x \in \llbracket [ \! \langle G \rangle \! ] \psi \rrbracket_{M_C}\) by the semantics.
    
    \emph{Case} \(\varphi = [\theta] [ \! \langle G \rangle \! ] \psi\). Suppose that \([\theta] [ \! \langle G \rangle \! ] \psi \in x\). Since \([\theta] [ \! \langle G \rangle \! ] \psi \in x\) is a necessity form and \(x\) is closed under \(R6\), this is equivalent to \(\forall \chi {\in}\mathcal{L}_{EL}^G\) \ \(\exists \tau {\in} \mathcal{L}_{EL}^{A \setminus G}\): \([\theta] (\chi \rightarrow \langle \chi \wedge \tau \rangle \psi) \in x\). 
    By the fact that \([\theta] (\chi \rightarrow \langle \chi \wedge \tau \rangle \psi) <^{Size}_{[], [ \! \langle \! \rangle \! ]} [\theta][ \! \langle G \rangle \! ] \psi\), the latter holds if and only if \(\forall \chi {\in} \mathcal{L}_{EL}^G\) \ \(\exists \tau {\in} \mathcal{L}_{EL}^{A \setminus G}\): \(x \in \llbracket [\theta] (\chi \rightarrow \langle \chi \wedge \tau \rangle \psi) \rrbracket_{M_C}\), which is equivalent to \(x \in \llbracket [\theta] [ \! \langle G \rangle \! ] \psi \rrbracket_{M_C}\) by the semantics.
\end{proof}

Proposition \ref{prop::ph} implies the following fact.

\begin{proposition}
\label{prop::truth}
    Let \(\varphi \in \mathcal{L}_{CoGAL}\), and \(x\) be a maximal consistent theory. Then \(\varphi \in x\) iff \(x \in \llbracket \varphi \rrbracket_{M_C}\). 
\end{proposition}

Finally, we prove the completeness of \(\mathbf{CoGAL}\).

\begin{proposition}
\label{prop::completeness}
    For all \(\varphi \in \mathcal{L}_{CoGAL}\), if \(\varphi\) is valid, then \(\varphi \in \mathbf{CoGAL}\).
\end{proposition}

\begin{proof}
    Towards a contradiction, suppose that \(\varphi\) is valid and \(\varphi \not \in \mathbf{CoGAL}\). Since \(\mathbf{CoGAL}\) is a consistent theory, and by Propositions \ref{prop::mct} and \ref{prop::consistency}, we have that \(\mathbf{CoGAL} + \neg \varphi\) is a consistent theory. Then, by Proposition \ref{prop::lindenbaum}, there exists a maximal consistent theory \(x \supseteq \mathbf{CoGAL} + \neg \varphi\), such that \(\neg \varphi \in x\). 
    By Proposition \ref{prop::truth}, this means that \(x \not \in \llbracket \varphi \rrbracket_{M_C}\), which contradicts \(\varphi\) being a validity.
\end{proof}

\section{Conclusion} \label{sec:conclusion}

We presented \(\mathbf{CoGAL}\) and provided a complete axiomatisation for it. 
The proof of completeness hinges on the validity of the axiom \(\langle \! [ G ] \! \rangle \varphi \rightarrow \langle G \rangle [ A \setminus G ] \varphi\). Validity of the other direction of the axiom, however, is still an open question. Answering it either way, positively, or negatively, will allow us to understand better mutual expressivity of \(\mathbf{CAL}\) and \(\mathbf{GAL}\).
The axiomatisation of \(\mathbf{CoGAL}\) we presented is infinitary and employs necessity forms. Finding a finitary axiomatisation is yet another open problem. 
An interesting avenue of further research is adding common and distributed knowledge operators to \(\mathbf{CoGAL}\) in the vein of \cite{agotnesA12}. 
Additionally, since it is known that \(\mathbf{GAL}\), \(\mathbf{CAL}\) \cite{agotnes16}, and hence \(\mathbf{CoGAL}\), are undecidable, a search for decidable fragments of these logics is another research question. We would also like to investigate applicability of logics with group and coalition announcements to epistemic planning \cite{bolander11}.
Finally, a complete axiomatisation of \(\mathbf{CAL}\) without group announcement operators has not been provided yet, and it is an intriguing direction of further research. 

\section*{Acknowledgements}

We would like to thank three anonymous reviewers for their insightful suggestions and detailed comments. 

\bibliographystyle{eptcs}
\bibliography{cal}

\appendix

\section{Coalition and Group Annoucement Logic Subsumes Coalition Logic}

\begin{definition}

Axiomatisation of \(\mathbf{CL}\) is as follows:

\(
\begin{array}[h]{ll}
(C0) &\textrm{all instantiation of propositional tautologies},\\

(C1) &\neg \langle \! [ G ] \! \rangle \bot, \\

(C2) &\langle \! [ G ] \! \rangle \top, \\

(C3) &\neg \langle \! [ \emptyset ] \! \rangle \neg \varphi \rightarrow \langle \! [ A ] \! \rangle \varphi,

\end{array}\)
\(
\begin{array}[h]{ll}

(C4) &\langle \! [ G ] \! \rangle (\varphi \wedge \psi) \rightarrow \langle \! [ G ] \! \rangle \varphi, \\      

(C5) &\langle \! [ G ] \! \rangle \varphi \wedge\langle \! [ H ] \! \rangle \psi \rightarrow \langle \! [ G \cup H ] \! \rangle (\varphi \wedge \psi), \\
	&\qquad \qquad \qquad \quad \textrm{ if } G \cap H = \emptyset, \\

(R0) &\vdash \varphi, \varphi \rightarrow \psi \Rightarrow \, \vdash \psi,\\

(R1) & \vdash \varphi \leftrightarrow \psi \Rightarrow \, \vdash \langle \! [ G ] \! \rangle \varphi \leftrightarrow \langle \! [ G ] \! \rangle \psi.

\end{array}\)

\end{definition}

\begin{proposition}

\(\mathbf{CoGAL}\) contains \(\mathbf{CL}\).

\end{proposition}

\begin{proof}

\(C0\) and \(R0\) are already in \(\mathbf{CoGAL}\).

\(R1\): \(\vdash \varphi \leftrightarrow \psi \Rightarrow \, \vdash \langle \! [ G ] \! \rangle \varphi \leftrightarrow \langle \! [ G ] \! \rangle \psi\). Assume that\(\models \varphi \leftrightarrow \psi\). This means that for any pointed model \((M,w)\) the following holds: \((M,w) \models \varphi\) iff \((M,w) \models \psi\) (1). Now, suppose that for some pointed model \((M,v)\) it holds that \((M,v) \models \langle \! [ G ] \! \rangle \varphi\). By the semantics, \(\exists \chi {\in} \mathcal{L}_{EL}^{G}\) \ \(\forall \tau {\in} \mathcal{L}_{EL}^{A \setminus G}\): \((M,v) \models \chi \wedge [\chi \wedge \tau] \varphi\), which is equivalent to the following: \((M,v) \models \chi\) and (\((M,v) \models \chi \wedge \tau\) implies \((M^{\chi \wedge \tau},v) \models \varphi\)). By (1), we have that \(\exists \chi {\in} \mathcal{L}_{EL}^{G}\) \ \(\forall \tau {\in} \mathcal{L}_{EL}^{A \setminus G}\): \((M,v) \models \chi\) and (\((M,v) \models \chi \wedge \tau\) implies \((M^{\chi \wedge \tau},v) \models \psi\)), which is \((M,v) \models \langle \! [ G ] \! \rangle \psi\) by the semantics. The same argumet holds in the other direction.

\(C1\): \(\neg \langle \! [ G  ] \! \rangle \bot\). From \(\top\) we derive \([ \! \langle G \rangle \! ] \top\) by \(R4\). Using the dual of box, we have \(\neg \langle \! [ G ] \! \rangle \neg \top\), or \(\neg \langle \! [ G ] \! \rangle \bot\).

\(C2\): \(\langle \! [ G ] \! \rangle \top\). In any state, there exists a true
announcement by $G$ (each agent in $G$ announces their knowledge of a tautology) and after any joint announcement, $\top$ is true, hence, the axiom is valid. 
%First, observe that \([ \! \langle G \rangle \! ] \bot\) is unsatisfiable, that is always false, since agent from \(G\) can always announce \(\top\). Then, the following is valid: \([ \! \langle G \rangle \! ] \bot \rightarrow \bot\). By contraposition, we have \(\neg \bot \rightarrow \neg ([ \! \langle G \rangle \! ] \bot\), and then \(\top \rightarrow \langle \! [ G ] \! \rangle \neg \bot\), which is \(\top \rightarrow \langle \! [ G ] \! \rangle \top\).

\(C3\): \(\neg \langle \! [ \emptyset ] \! \rangle \neg \varphi \rightarrow \langle \! [ A ] \! \rangle \varphi\). By the semantics, \(\neg \langle \! [ \emptyset ] \! \rangle \neg \varphi\), which is \([ \! \langle \emptyset \rangle \! ] \varphi\), means that there exists some \(\psi \in \mathcal{L}_{EL}^{A}\), such that \(\langle \psi \rangle \varphi\). This is precisely the meaning of \(\langle \! [ A ] \! \rangle \varphi\).

\(C4\): \(\langle \! [ G ] \! \rangle (\varphi \wedge \psi) \rightarrow \langle \! [ G ] \! \rangle \varphi\).  Suppose that  \(\langle \! [ G ] \! \rangle (\varphi \wedge \psi) \) holds. By the semantics, \(\exists \chi {\in} \mathcal{L}_{EL}^{G}\) \ \(\forall \tau {\in} \mathcal{L}_{EL}^{A \setminus G}\): \(\chi \wedge [\chi \wedge \tau] (\varphi \wedge \psi)\). Then, by \(A7\), we have \(\exists \chi {\in} \mathcal{L}_{EL}^{G}\) \ \(\forall \tau {\in} \mathcal{L}_{EL}^{A \setminus G}\): \(\chi \wedge [\chi \wedge \tau] \varphi \wedge  [\chi \wedge \tau] \psi\). The latter implies  \(\exists \chi {\in} \mathcal{L}_{EL}^{G}\) \ \(\forall \tau {\in} \mathcal{L}_{EL}^{A \setminus G}\): \(\chi \wedge [\chi \wedge \tau] \varphi\), which is \(\langle \! [ G ] \! \rangle \varphi\) by the semantics.

\(C5\): \(\langle \! [ G ] \! \rangle \varphi_1 \wedge\langle \! [ H ] \! \rangle \varphi_2 \rightarrow \langle \! [ G \cup H ] \! \rangle (\varphi_1 \wedge \varphi_2),	\textrm{ if } G \cap H = \emptyset\).  Assume that \(\langle \! [ G ] \! \rangle \varphi_1 \wedge\langle \! [ H ] \! \rangle \varphi_2\) holds. Let us consider the first conjunct. By the semantics, we have \(\exists \psi_G {\in} \mathcal{L}_{EL}^{G}\) \ \(\forall \chi_{A \setminus G} {\in} \mathcal{L}_{EL}^{A \setminus G}\): \(\psi_G \wedge [\psi_G \wedge \chi_{A \setminus G}]\varphi_1\) (1).  The latter implies \(\exists \psi_{G \cup H} {\in} \mathcal{L}_{EL}^{G \cup H}\) \ \(\forall \chi_{A \setminus {G \cup H}} {\in} \mathcal{L}_{EL}^{A \setminus {G \cup H}}\): \(\psi_{G \cup H} \wedge [\psi_{G \cup H} \wedge \chi_{A \setminus G \cup H}]\varphi_1\). In order to show this let us assume to the contrary that  \(\forall \psi_{G \cup H} {\in} \mathcal{L}_{EL}^{G \cup H}\) \ \(\exists \chi_{A \setminus {G \cup H}} {\in} \mathcal{L}_{EL}^{A \setminus {G \cup H}}\): \(\psi_{G} \rightarrow \langle \psi_{G} \wedge \psi_H \wedge \chi_{A \setminus G \cup H} \rangle \neg \varphi_1\). Next, by (1), we fix some \(\psi_G\). Then, we have \(\forall \psi_{H} {\in} \mathcal{L}_{EL}^{H}\) \ \(\exists \chi_{A \setminus {G \cup H}} {\in} \mathcal{L}_{EL}^{A \setminus {G \cup H}}\): \(\psi_{G} \rightarrow \langle \psi_{G} \wedge \psi_H \wedge \chi_{A \setminus G \cup H} \rangle \neg \varphi_1\), which means that there is a combination of \(\psi_H\) and \(\chi_{A \setminus G \cup H}\) that makes \(\varphi_1\) false. Since sets \(H\) and \(A \setminus G \cup H\) comprise \(A \setminus G\) (on condition that \(G \cap H = \emptyset\)), this means that there is some \(\chi_{A \setminus G}\) which enforces \(\neg \varphi_1\). Hence, the contradiction with (1).
Similarly, we can prove that \(\exists \psi_{G \cup H} {\in} \mathcal{L}_{EL}^{G \cup H}\) \ \(\forall \chi_{A \setminus {G \cup H}} {\in} \mathcal{L}_{EL}^{A \setminus {G \cup H}}\): \(\psi_{G \cup H} \wedge [\psi_{G \cup H} \wedge \chi_{A \setminus G \cup H}]\varphi_2\). 
By distributivity of $[\ ]$ over $\wedge$,
we get \(\exists \psi_{G \cup H} {\in} \mathcal{L}_{EL}^{G \cup H}\) \ \(\forall \chi_{A \setminus {G \cup H}} {\in} \mathcal{L}_{EL}^{A \setminus {G \cup H}}\): \(\psi_{G \cup H} \wedge [\psi_{G \cup H} \wedge \chi_{A \setminus G \cup H}]\ (\varphi_1 \wedge \varphi_2)\). 
Hence $\langle \! [ G \cup H ] \! \rangle (\varphi_1 \wedge \varphi_2)$.
\end{proof}

\section{Proofs of Propositions}

\begin{proposition}
\label{prop::app1}
    \(R4\) and \(R6\) are sound, that is, they preserve validity.
\end{proposition}
\begin{proof}
  (\(R4\)) Assume that \(\models \varphi\). By \(R2\), for an arbitrary \(\psi\), \(\models [\psi] \varphi\). Since \(\psi\) is arbitrary, \(\models [ \! \langle G \rangle \! ] \varphi\); in other words, whatever agents announce, they cannot make a valid formula false.
    
    (\(R6\)) Let \((M,w)\) be an arbitrary pointed model. We proceed by induction on \(\eta\). 

    \emph{Base case}. \(\forall \psi {\in}  \mathcal{L}_{EL}^G\) \(\exists \chi {\in}  \mathcal{L}_{EL}^{A \setminus G}\): \(\psi \rightarrow \langle \psi \wedge \chi \rangle \varphi\) is valid. Therefore, by the semantics, we infer validity of \([ \! \langle G \rangle \! ] \varphi\). 

    \emph{Induction hypothesis}. %Natasha: added induction hypothesis.
Assume the rule preserves validity for $n(\eta(\psi \rightarrow \langle \psi \wedge \chi \rangle \varphi) = k$. 
We show that it holds for $k+1$. 
% Natasha: I am not sure what the next statement is doing here? Was it meant to say the same as above?
%\(\forall \psi{\in}  \mathcal{L}_{EL}^G\) \(\exists \chi{\in}  \mathcal{L}_{EL}^{A \setminus G}\): \(\eta (\psi \rightarrow \langle \psi \wedge \chi \rangle \varphi)\) is valid, and, consequently, \([ \! \langle G \rangle \! ] \varphi\) is valid as well. 

    \emph{Case} \(\forall \psi {\in}  \mathcal{L}_{EL}^G\) \(\exists \chi {\in}  \mathcal{L}_{EL}^{A \setminus G}\): \(\tau \rightarrow \eta (\psi \rightarrow \langle \psi \wedge \chi \rangle \varphi)\) is valid. This means that \((M,w) \models \tau \rightarrow \eta (\psi \rightarrow \langle \psi \wedge \chi \rangle \varphi)\) iff \((M,w) \models \neg \tau\) or \((M,w) \models \eta (\psi \rightarrow \langle \psi \wedge \chi \rangle \varphi)\), which is \((M,w) \models \neg \tau\) or \((M,w) \models [ \! \langle G \rangle \! ] \varphi\) by Induction hypothesis. Hence, \((M,w) \models \tau \rightarrow [ \! \langle G \rangle \! ] \varphi\). 

    \emph{Case} \(\forall \psi{\in}  \mathcal{L}_{EL}^G\) \(\exists \chi {\in}  \mathcal{L}_{EL}^{A \setminus G}\): \(K_a \eta (\psi \rightarrow \langle \psi \wedge \chi \rangle \varphi)\) is valid. This means that \((M,w) \models K_a \eta (\psi \rightarrow \langle \psi \wedge \chi \rangle \varphi)\). By the semantics, for every \(v \in W : (w,v) \in \sim_a\) implies \((M,v) \models \eta (\psi \rightarrow \langle \psi \wedge \chi \rangle \varphi)\). By Induction hypothesis, for every \(v \in W : (w,v) \in \sim_a \) implies \((M,v) \models \eta ([ \! \langle G \rangle \! ] \varphi)\). And, by the semantics, \((M,w) \models K_a \eta ([ \! \langle G \rangle \! ] \varphi)\).

    \emph{Case} \(\forall \psi {\in}  \mathcal{L}_{EL}^G\) \(\exists \chi {\in}  \mathcal{L}_{EL}^{A \setminus G}\): \([\tau]\eta (\psi \rightarrow \langle \psi \wedge \chi \rangle \varphi)\) is valid. This means that \((M,w) \models [\tau]\eta (\psi \rightarrow \langle \psi \wedge \chi \rangle \varphi)\). By the semantics, \((M,w) \models \tau\) implies \((M^\tau, w) \models \eta (\psi \rightarrow \langle \psi \wedge \chi \rangle \varphi)\). By Induction hypothesis, \((M,w) \models \tau\) implies \((M^\tau, w) \models [ \! \langle G \rangle \! ] \varphi\). Finally, by the semantics, \((M,w) \models [\tau ][ \! \langle G \rangle \! ] \varphi\).
\end{proof}

\begin{proposition}
\label{prop::app2}
    \(\langle \! [ G ] \! \rangle \langle \! [ H ] \! \rangle \varphi  \rightarrow \langle \! [ G \cup H ] \! \rangle \varphi\) is valid.
\end{proposition}
\begin{proof}
Let \(\psi:=  \bigwedge_{i \in G} K_i \psi_i\), \(\psi^\prime:=  \bigwedge_{j \in A \setminus G}\) \(K_j \psi_j^\prime\), \(\chi:= \bigwedge_{k \in H} K_k \chi_k\), and \(\chi^\prime:=\bigwedge_{l \in A \setminus H} K_l \chi_l^\prime\).
    Suppose \((M,w) \models \langle \! [ G ] \! \rangle \langle \! [ H ] \! \rangle \varphi\) for some \(M\) and \(w \in W\). 
    By the semantics, 
\begin{align*}
\exists \psi {\in} \mathcal{L}_{EL}^{G} \ \forall \psi^\prime {\in} \mathcal{L}_{EL}^{A \setminus G} \ \exists \chi {\in} \mathcal{L}_{EL}^{H} \ \forall \chi^\prime {\in} \mathcal{L}_{EL}^{A \setminus H}: (M,w) \models \psi \wedge [\psi \wedge \psi^\prime] (\chi \wedge [\chi \wedge \chi^\prime] \varphi). 
\end{align*}
%\(\exists \bigwedge_{i \in G} K_i \psi_i\), \(\forall \bigwedge_{j \in A \setminus G}\) \(K_j \psi_j^\prime\), \(\exists \bigwedge_{k \in H} K_k \chi_k\), \(\forall \bigwedge_{l \in A \setminus H} K_l \chi_l^\prime\): 
%\begin{align*}
	%&(M,w) \models   \bigwedge_{i \in G} K_i \psi_i \wedge [ \bigwedge_{i \in G} K_i \psi_i \wedge \bigwedge_{j \in A \setminus G} K_j \psi_j^\prime ] ( \bigwedge_{k \in H} K_k \chi_k \wedge\\ 
%	&\quad \wedge [ \bigwedge_{k \in H} K_k \chi_k \wedge \bigwedge_{l \in A \setminus H} K_l \chi_l^\prime ] ) \varphi. 
%\end{align*}
By \(A7\), we have: \((M,w) \models \psi \wedge [\psi \wedge \psi^\prime] \chi \wedge  [\psi \wedge \psi^\prime] [\chi \wedge \chi^\prime] \varphi\). 
%\begin{align*}
%	&(M,w) \models   \bigwedge_{i \in G} K_i \psi_i \wedge [ \bigwedge_{i \in G} K_i \psi_i \wedge \bigwedge_{j \in A \setminus G} K_j \psi_j^\prime ] \bigwedge_{k \in H} K_k \chi_k \wedge\\ 
%	&\quad \wedge [ \bigwedge_{i \in G} K_i \psi_i \wedge \bigwedge_{j \in A \setminus G} K_j \psi_j^\prime ] [ \bigwedge_{k \in H} K_k \chi_k \wedge \bigwedge_{l \in A \setminus H} K_l \chi_l^\prime ] \varphi. 
%\end{align*}
We are interested now in the third conjunct: \( [\psi \wedge \psi^\prime] [\chi \wedge \chi^\prime] \varphi\). %\((M,w) \models  [ \bigwedge_{i \in G} K_i \psi_i \wedge \bigwedge_{j \in A \setminus G} K_j \psi_j^\prime ] [ \bigwedge_{k \in H} K_k \chi_k \wedge \bigwedge_{l \in A \setminus H} K_l \chi_l^\prime ] \varphi\). 
By \(A9\), we have that \((M,w) \models  [\psi \wedge \psi^\prime \wedge [\psi \wedge \psi^\prime] (\chi \wedge \chi^\prime)] \varphi\).
%\(\exists \bigwedge_{i \in G}\) \(K_i \psi_i\), \(\forall \bigwedge_{j \in A \setminus G}\) \(K_j \psi_j^\prime\), \(\exists \bigwedge_{k \in H}\) \(K_k \chi_k\), \(\forall \bigwedge_{l \in A \setminus H}\) \(K_l \chi_l^\prime\):
%\begin{align*}
  %          &(M,w) \models [ \bigwedge_{i \in G} K_i \psi_i \wedge \bigwedge_{j \in A \setminus G} K_j \psi^\prime_j \wedge [\bigwedge_{i \in G} K_i \psi_i \wedge \bigwedge_{j \in A \setminus G} K_j \psi_j^\prime] (\bigwedge_{k \in H} K_k \chi_k \wedge \bigwedge_{l \in A \setminus H} K_l \chi_l^\prime) ] \varphi.
%\end{align*}
    Now, let us examine the following conjunction: \(\exists \psi {\in} \mathcal{L}_{EL}^{G}\) \ \(\forall \psi^\prime {\in} \mathcal{L}_{EL}^{A \setminus G}\): \(\psi \wedge \psi^\prime\), which is 
\begin{align*}
\exists \bigwedge_{i \in G} K_i \psi_i \quad \forall \bigwedge_{j \in A \setminus G} K_j \psi_j^\prime: \bigwedge_{i \in G} K_i \psi_i \wedge \bigwedge_{j \in A \setminus G} K_j \psi_j^\prime
\end{align*}
 in the full form. 
We can present the set of agents \(A \setminus G\) as a union of \(G \cup H\) and \(H \setminus G\) by expanding the right conjunct.
So, we have \(\exists \bigwedge_{i \in G}\) \(K_i \psi_i\) \ \(\forall \bigwedge_{m \in H \setminus G}\) \(K_m \psi_m^{\prime \prime}\) \ \(\forall \bigwedge_{j \in A \setminus G \cup H}\) \(K_j \psi_j^\prime\): 
    \begin{equation}
    \label{eq::1}
        \bigwedge_{i \in G} K_i \psi_i \wedge \bigwedge_{m \in H \setminus G} K_m \psi_m^{\prime \prime} \wedge \bigwedge_{j \in A \setminus G \cup H} K_j \psi_j^\prime.
    \end{equation}
 Since none of the universal quantifiers here is vacuous, there are particular \(\psi^{\prime \prime}\) for which the conjunction holds. Formally, \(\exists \bigwedge_{i \in G}\) \(K_i \psi_i\) \ \(\exists \bigwedge_{m \in H \setminus G}\) \(K_m \psi_m^{\prime \prime}\) \ \(\forall \bigwedge_{j \in A \setminus G \cup H}\) \(K_j \psi_j^\prime\): (\ref{eq::1}). Therefore, combining \(G\) and \(H \setminus G\),  we have 
\begin{align*}
\exists \bigwedge_{i \in G \cup H} K_i \psi_i \quad \forall \bigwedge_{j \in A \setminus G \cup H} K_j \psi_j^\prime: \bigwedge_{i \in G \cup H} K_i \psi_i \wedge \bigwedge_{j \in A \setminus G \cup H} K_j \psi_j^\prime.
\end{align*}
The same argument holds for the conjunction \(\exists \chi {\in} \mathcal{L}_{EL}^{H}\) \ \(\forall \chi^\prime {\in} \mathcal{L}_{EL}^{G}\): \(\chi \wedge \chi^\prime\).
Let us redefine our auxiliary formulae: \(\psi:=  \bigwedge_{i \in G \cup H} K_i \psi_i\), \(\psi^\prime:=  \bigwedge_{j \in A \setminus {G \cup H}}\) \(K_j \psi_j^\prime\), \(\chi:= \bigwedge_{i \in G \cup H} K_k \chi_k\), and \(\chi^\prime:=\bigwedge_{j \in A \setminus {G \cup H}} K_l \chi_l^\prime\).
    Thus, we have that \(\exists \psi {\in} \mathcal{L}_{EL}^{G \cup H}\) \ \(\forall \psi^\prime {\in} \mathcal{L}_{EL}^{A \setminus {G \cup H}}\) \ \(\exists \chi {\in}  \mathcal{L}_{EL}^{G \cup H}\) \ \(\forall \chi^\prime {\in} \mathcal{L}_{EL}^{A \setminus {G \cup H}}\): \((M,w) \models [\psi \wedge \psi^\prime \wedge [\psi \wedge \psi^\prime](\chi \wedge \chi^\prime)]\varphi\). In the full form, the latter is 
    \begin{align*}
&\exists  \bigwedge_{i \in G \cup H} \! K_i \psi_i \quad  \forall \! \bigwedge_{j \in A \setminus G \cup H} \! K_j \psi_j^\prime \quad \exists \bigwedge_{i \in G \cup H} K_i \chi_i \quad \forall \bigwedge_{j \in A \setminus G \cup H} K_j \chi_j^\prime: \\
         &(M,w) \models [ \bigwedge_{i \in G \cup H} K_i \psi_i \wedge \bigwedge_{j \in A \setminus G \cup H} K_j \psi^\prime_j \wedge [\bigwedge_{i \in G \cup H} K_i \psi_i \wedge \bigwedge_{j \in A \setminus G \cup H} K_j \psi^\prime_j] (\bigwedge_{i \in G \cup H} K_i \chi_i \wedge \bigwedge_{j \in A \setminus G \cup H} K_j \chi_j^\prime) ] \varphi.   
    \end{align*}
   Using \(A7\) and \(A8\), we can `push' announcements into the scope of knowledge operators:  
    \begin{align*}
           & (M,w) \models [\psi \wedge  \psi^\prime \wedge (\bigwedge_{i \in G \cup H} \bigwedge_{j \in A \setminus G \cup H} (K_i \psi_i \wedge K_j \psi^\prime_j \rightarrow K_i [ K_i \psi_i \wedge K_j \psi^\prime_j] \chi_i)  \wedge \\ 
& \hspace{8cm} \wedge ( K_i \psi_i \wedge K_j \psi^\prime_j \rightarrow K_j [K_i \psi_i \wedge K_j \psi^\prime_j] \chi_j^\prime)) ] \varphi.
    \end{align*}
    By propositional reasoning, the latter is equivalent to \((M,w) \models\) \([ \bigwedge_{i \in G \cup H}\) \(\bigwedge_{j \in A \setminus G \cup H}\) \((K_i \psi_i\) \(\wedge\) \(K_j \psi^\prime_j\) \(\wedge\) \(K_i [ K_i \psi_i\) \(\wedge\) \(K_j \psi^\prime_j] \chi_i\) \(\wedge\) \(K_j [K_i \psi_i\) \(\wedge\)  \(K_j \psi^\prime_j]\) \(\chi_j^\prime) ] \varphi\). Finally, we have 
\begin{align*}
&\exists \bigwedge_{i \in G \cup H} K_i \psi_i \quad \forall \bigwedge_{j \in A \setminus G \cup H} K_j \psi_j^\prime \quad \exists \bigwedge_{i \in G \cup H} K_i \chi_i \quad \forall \bigwedge_{j \in A \setminus G \cup H} K_j \chi_j^\prime:\\ 
&(M,w) \models [ \bigwedge_{i \in G \cup H} \bigwedge_{j \in A \setminus G \cup H} (K_i (\psi_i\wedge [ K_i \psi_i \wedge K_j \psi^\prime_j] \chi_i) \wedge K_j (\psi^\prime_j \wedge [K_i \psi_i \wedge K_j \psi^\prime_j] \chi_j^\prime)) ] \varphi.
\end{align*}    

    Conjuncts of the form \(K_i (\psi_i\) \(\wedge\) \([ K_i \psi_i\) \(\wedge\) \(K_j \psi^\prime_j] \chi_i)\) mean that agent \(i\) can announce \(\psi_i\), i.e. what she knows now, or \([ K_i \psi_i\) \(\wedge\) \(K_j \psi^\prime_j] \chi_i\) (which is equivalent to \(t([ K_i \psi_i\) \(\wedge\) \(K_j \psi^\prime_j] \chi_i)\)), i.e. what she will know after announcements of other agents but not necessarily knows now, or both.
    Since all the variants comprise \(\mathcal{L}_{EL}^{G \cup H}\), we rewrite the notation. Hence, \(\exists \bigwedge_{i \in G \cup H}\) \(K_i \tau_i\) \ \(\forall \bigwedge_{j \in A \setminus G \cup H}\) \(K_j \tau_j^\prime\): \((M,w) \models\) \([ \bigwedge_{i \in G \cup H}\) \(\bigwedge_{j \in A \setminus G \cup H}\) \((K_i \tau_i\) \(\wedge\) \(K_j \tau^\prime_j ) ] \varphi\), and  at the same time \((M,w) \models \bigwedge_{i \in G \cup H} K_i \tau_i\) (\(\bigwedge_{i \in G} K_i \psi_i\) is equivalent to \(\bigwedge_{i \in G \cup H} K_i \psi_i\), where agents from \(H\) announce \(\top\)).
    And, by the semantics, this is \((M,w) \models \langle \! [ G \cup H ] \! \rangle \varphi\).
\end{proof}

\begin {proposition}
\label{prop::app3}
    Let \(x\) be a theory, \(\varphi, \psi \in \mathcal{L}_{CoGAL}\), and \(a \in A\). The following are theories: \(x + \varphi = \{\psi: \varphi \rightarrow \psi \in x\}, K_a x = \{\varphi: K_a \varphi \in x\}\), and \([\varphi]x = \{\psi: [\varphi]\psi \in x\}\).
\end{proposition}

\begin{proof}
    We just expand the proof from \cite {balbiani08} by showing that corresponding theories are closed under \(R5\) and \(R6\).
    
    Suppose that \(\eta ([\psi]\chi) \in x + \varphi\) for all \(\psi \in\mathcal{L}_{EL}^G\). 
    It means that \(\varphi \rightarrow \eta ([\psi]\chi) \in x\) for all \(\psi \in \mathcal{L}_{EL}^G\). 
    Since \(\varphi \rightarrow \eta ([\psi]\chi)\) is a necessity form, and \(x\) is closed under \(R5\), we infer that \(\varphi \rightarrow \eta ([G]\chi) \in x\), and, consequently, \(\eta ([G]\chi) \in x + \varphi\). 
    So, \(x + \varphi\) is closed under \(R5\).
    Now, let \(\forall \psi {\in} \mathcal{L}_{EL}^G\) \ \(\exists \tau {\in} \mathcal{L}_{EL}^{A \setminus G}:\) \(\eta( \psi \rightarrow \langle \psi \wedge \tau \rangle \chi) \in x + \varphi\). 
    It means that  \(\forall \psi {\in} \mathcal{L}_{EL}^G\) \ \(\exists \tau {\in} \mathcal{L}_{EL}^{A \setminus G}:\) \(\varphi \rightarrow \eta( \psi \rightarrow \langle \psi \wedge \tau \rangle \chi) \in x\). 
    Since \(\varphi \rightarrow \eta( \psi \rightarrow \langle \psi \wedge \tau \rangle \chi)\) is a necessity form, and \(x\) is closed under \(R6\), we infer that \(\varphi \rightarrow \eta( [ \! \langle G \rangle \! ] \chi) \in x\), and, consequently, \(\eta( [ \! \langle G \rangle \! ] \chi) \in x + \varphi\). So, \(x + \varphi\) is closed under \(R6\).
    
    Suppose that \(\eta ([\psi]\chi) \in K_a x\) for all \(\psi \in\mathcal{L}_{EL}^G\). 
    It means that \(K_a \eta ([\psi]\chi) \in x\) for all \(\psi \in \mathcal{L}_{EL}^G\). 
    Since \(K_a \eta ([\psi]\chi)\) is a necessity form, and \(x\) is closed under \(R5\), we infer that \(K_a \eta ([G]\chi) \in x\), and, consequently, \(\eta ([G]\chi) \in K_a x\). 
    So, \(K_a x\) is closed under \(R5\).
    Now, let  \(\forall \psi {\in} \mathcal{L}_{EL}^G\) \ \(\exists \tau {\in}\mathcal{L}_{EL}^{A \setminus G}:\) \(\eta ( \psi \rightarrow \langle \psi \wedge \tau \rangle \chi) \in K_a x\). 
    It means that  \(\forall \psi {\in} \mathcal{L}_{EL}^G\) \ \(\exists \tau {\in} \mathcal{L}_{EL}^{A \setminus G}:\) \(K_a \eta( \psi \rightarrow \langle \psi \wedge \tau \rangle \chi) \in x\). Since \(K_a \eta( \psi \rightarrow \langle \psi \wedge \tau \rangle \chi)\) is a necessity form, and \(x\) is closed under \(R6\), we infer that \(K_a \eta( [ \! \langle G \rangle \! ] \chi) \in x\), and, consequently, \(\eta( [ \! \langle G \rangle \! ] \chi) \in K_a x\). So, \(K_a x\) is closed under \(R6\).
    
    Finally, suppose that \(\eta ([\psi]\chi) \in [\varphi] x\) for all \(\psi \in \mathcal{L}_{EL}^G\). 
    It means that \([\varphi] \eta ([\psi]\chi) \in x\) for all \(\psi \in\mathcal{L}_{EL}^G\). 
    Since \([\varphi] \eta ([\psi]\chi)\) is a necessity form, and \(x\) is closed under \(R5\), we infer that \([\varphi] \eta ([G]\chi) \in x\), and, consequently, \(\eta ([G]\chi) \in [\varphi] x\). 
    So, \([\varphi] x\) is closed under \(R5\).
    Now, let \(\forall \psi \in \mathcal{L}_{EL}^G\) \ \(\exists \tau \in \mathcal{L}_{EL}^{A \setminus G}:\) \(\eta( \psi \rightarrow \langle \psi \wedge \tau \rangle \chi) \in [\varphi] x\). 
    It means that  \(\forall \psi {\in} \mathcal{L}_{EL}^G\) \ \(\exists \tau {\in} \mathcal{L}_{EL}^{A \setminus G}:\) \([\varphi] \eta( \psi \rightarrow \langle \psi \wedge \tau \rangle \chi) \in x\). 
    Since \([\varphi] \eta( \psi \rightarrow \langle \psi \wedge \tau \rangle \chi)\) is a necessity form, and \(x\) is closed under \(R6\), we infer that \([\varphi] \eta( [ \! \langle G \rangle \! ] \chi) \in x\), and, consequently, \(\eta ( [ \! \langle G \rangle \! ] \chi) \in [\varphi] x\). So, \([\varphi] x\) is closed under \(R6\).
\end{proof}

\begin{proposition}
\label{prop::balbiani}
Let \(\varphi \in \mathcal{L}_{CoGAL}\). Then \(\mathbf{CoGAL} + \varphi\) is consistent iff \(\neg \varphi \not \in \mathbf{CoGAL}\).
\end{proposition}

\begin{proof}

  \textit{From left to right}. Suppose to the contrary that \(\mathbf{CoGAL} + \varphi\) is consistent and \(\neg \varphi \in \mathbf{CoGAL}\). Then, having both \(\varphi\) and \(\neg \varphi\) means that \(\bot \in \mathbf{CoGAL}+\varphi\), which contradicts to \(\mathbf{CoGAL}+\varphi\) being consistent. 
  
  \textit{From right to left}. Let us consider the contrapositive: if \(\mathbf{CoGAL} + \varphi\) is inconsistent, then \(\neg \varphi \in \mathbf{CoGAL}\). Since \(\mathbf{CoGAL} + \varphi\) is inconsistent, \(\bot \in \mathbf{CoGAL} + \varphi\), or, by Proposition \ref{prop::app3}, \(\varphi \rightarrow \bot \in \mathbf{CoGAL}\). By consistency of \(\mathbf{CoGAL}\) and propositional reasoning, we have that \(\neg \varphi \in \mathbf{CoGAL}\). \qedhere
\end{proof} 

\end{document}